\documentclass{ws-ijgmmp}
\usepackage{float}
\usepackage{subfig}
\begin{document}

\markboth{Authors' Names}
{Instructions for Typing Manuscripts (Paper's Title)}

%
\catchline{}{}{}{}{}
%

\title{The geometry of null-like disformal transformations}

\author{Iarley P.\, Lobo}
\address{Departamento de F\'isica, Universidade Federal de Lavras, Caixa Postal 3037, 37200-000 Lavras-MG, Brazil\\
Departamento de F\'isica, Universidade Federal da Para\'iba, C. Postal 5008, Jo\~ao Pessoa, PB 58051-970, Brazil\\
\email{lobofisica@gmail.com and iarley\_lobo@fisica.ufpb.br}}

\author{Gabriel G.\ Carvalho}
\address{Centro de Inform\'atica, Universidade Federal de Pernambuco, Recife, Pernambuco,  50740-560, Brazil\\
Departamento de Física, Universidade Federal Rural de Pernambuco, 52171-900 Recife, PE, Brazil         
\email{ggc5@cin.ufpe.br}}

\maketitle

\begin{history}
\received{(Day Month Year)}
\revised{(Day Month Year)}
\end{history}

\begin{abstract}
Motivated by the hindrance of defining metric tensors compatible with the underlying spinor structure, other than the ones obtained via a conformal transformation, we study how some geometric objects are affected by the action of a disformal transformation in the closest scenario possible: the disformal transformation in the direction of a null-like vector field. Subsequently, we analyze symmetry properties such as mutual geodesics and mutual Killing vectors, generalized Weyl transformations that leave the disformal relation invariant, and introduce the concept of disformal Killing vector fields. In most cases, we use the Schwarzschild metric, in the Kerr-Schild formulation, to verify our calculations and results. We also revisit the disformal operator using a Newman-Penrose basis to show that, in the null-like case, this operator is not diagonalizable.
\end{abstract}


\section{Introduction}
Effective metrics constitute an important tool for the description of physical phenomena, since they allow researchers to explore intermediary regimes of a given system using standard tools of differential geometry, instead of dealing with different mathematical structures derived from different fundamental theories. Just to mention few examples, we can refer to applications in Analogue Gravity (see \cite{Barcelo:2000tg,Barcelo:2005fc} and references therein). Another interesting application in quantum gravity can be found in \cite{Magueijo:2002xx,Assaniousssi:2014ota,Torrome:2015cga}.
\par
Among these effective metrics, here we consider those generated by the so called {\it disformal transformations}. They are, essentially, a generalization of conformal transformations, in which besides a conformal rescaling of a background metric, it is considered a rank-2 tensor field with properties that guarantee the invertibility of the disformal metric. Although one could define disformal transformations in the context of Riemannian geometry, the case of a pseudo-Riemannian geometry with Lorentzian signature is far more rich due to the changes in the causal structure and its applications to physics.
\par
Since its appearance, in the nineties due to Bekenstein's works \cite{beken1,Bekenstein:1992pj}, disformal transformations have had numerous applications. For instance TeVeS theories for MOND \cite{Bekenstein:2004ne}, bimetric theories of gravity \cite{Clifton:2011jh}, scalar \cite{Novello:2012wr} or scalar-tensor theories \cite{Koivisto:2012za,Ip:2015qsa,Sakstein:2014isa,Sakstein:2015jca} (including Mimetic \cite{Deruelle:2014zza,Arroja:2015wpa,Myrzakulov:2015kda,Sebastiani:2016ras} and Horndeski gravity \cite{Bettoni:2013diz,Zumalacarregui:2013pma,Gleyzes:2014qga}), disformal inflation \cite{Kaloper:2003yf}, chiral symmetry breaking \cite{Bittencourt:2014fpa}, anomalous magnetic moment for neutrinos \cite{Novello:2014yea}, disformal invariance of matter fields \cite{Falciano:2011rf,Goulart:2013laa,Bittencourt:2015ypa}, quantum gravity \cite{Carvalho:2015omv} and more on analog gravity \cite{Novello:2012pv,Novello:2011sh}.
\par
In Ref.\cite{Carvalho:2015omv}, besides demonstrating a link between disformal transformations and rainbow gravity, we further analyzed the mathematical properties of disformal transformations defined by time-like vectors (with respect to the background metric). We defined a disformal operator that acted on the (co-)tangent space, which worked as a decomposition of the disformal transformation on the metric tensor. This way, we managed to, in certain sense, explain the group operation rule of successive disformal transformations and solve some ambiguity issues. 
\par
In the present text, we continue to study the properties of these transformations, now using null-like vectors to define them (which we call null-like disformal transformations). We explicitly decompose the connection and the curvature, study invariant geodesics and Killing vectors, generalize Weyl transformations for the geometry induced a null-like disformal transformation, define the disformal Killing equation and study the disformal operator using a Newman-Penrose basis. We also further decompose the null-like disformal transformation by its action on the spinor space, which could serve as a laboratory for inducing a disformal vector and furthermore a disformal metric.
\par
This text is divided as follows: in section (\ref{preliminaries}) we fix the notation used throughout the text, revise the definition of a derivative operator and list some formulae concerning conformal transformations. In section (\ref{disformasse}), we study the geometry of null-like disformal transformations providing some generalizations of results valid in the conformal frame. In multiple examples we provide a test case to verify our results. In section (\ref{disfopsec}) we mention some relevant algebraic differences between the null and time-like disformal operators. Finally, in section (\ref{conclusion}), we conclude with some future perspectives. In \ref{app-disf}, we attempt to construct a disformal transformation in spinor space that yields a spacetime disformal transformation and it is shown that there is not a spin basis transformation that produces a spacetime disformal metric.


\section{Preliminaries: derivative operators, curvature and conformal transformations}\label{preliminaries}
Though we assume the reader to be familiarized with differential geometry, we start revising the definition of a derivative operator showing that any two derivative operators differ by a tensor. We shall use Penrose's abstract index notation \cite{wald84}, in which tensor equations with Latin indices are true tensor equations (i.e., valid in any coordinate system) and tensor equations with Greek indices represent equations valid on a given coordinate system.

\begin{definition}[Derivative operators] A {\it derivative operator} $D_a$ on a manifold $\cal M$ is a map which takes each $C^{r}$ tensor field of type $(k,l)$ to a tensor field of type $(k,l+1)$ and satisfies
\begin{enumerate}
\item Linearity: For all $A,B \in {\cal T}(k,l)$ and $\alpha,\beta \in \mathbb{R}$,
\begin{eqnarray}
D_{a} \left[\alpha A^{b_{1}\dots b_{k}}_{\ \ \ \ \ \ \ c_{1}\dots c_{l}} + \beta B^{b_{1}\dots b_{k}}_{\ \ \ \ \ \ \ c_{1}\dots c_{l}}\right] = \alpha D_{a}  A^{b_{1}\dots b_{k}}_{\ \ \ \ \ \ \ c_{1}\dots c_{l}} + \beta D_{a}  B^{b_{1}\dots b_{k}}_{\ \ \ \ \ \ \ c_{1}\dots c_{l}}.
\end{eqnarray}

\item Leibnitz rule: For all $A \in {\cal T}(k,l)$ and $B \in {\cal T}(k',l')$,
\begin{eqnarray}
D_{e}\left[ A^{a_{1}\dots a_{k}}_{\ \ \ \ \ \ \ b_{1}\dots b_{l}}  B^{c_{1}\dots c_{k'}}_{\ \ \ \ \ \ \ d_{1}\dots d_{l'}}\right]  &=& D_{e}\left[ A^{a_{1}\dots a_{k}}_{\ \ \ \ \ \ \ b_{1}\dots b_{l}} \right] B^{c_{1}\dots c_{k'}}_{\ \ \ \ \ \ \ d_{1}\dots d_{l'}}\nonumber\\ &+& A^{a_{1}\dots a_{k}}_{\ \ \ \ \ \ \ b_{1}\dots b_{l}} D_{e}\left[  B^{c_{1}\dots c_{k'}}_{\ \ \ \ \ \ \ d_{1}\dots d_{l'}}\right] .
\end{eqnarray}
\item Commutativity with contraction: For all $A \in {\cal T}(k,l)$,
\begin{eqnarray}
D_{d} \left[ A^{a_{1}\dots c\dots a_{k}}_{\ \ \ \ \ \ \ \ \ \ b_{1}\dots c\dots b_{l}} \right] = D_{d} A^{a_{1}\dots c\dots a_{k}}_{\ \ \ \ \ \ \ \ \ \ b_{1}\dots c\dots b_{l}}.
\end{eqnarray}

\item Consistency with the notion of vector fields as directional derivatives on scalar fields: For all $f \in C^{r}(\cal M)$ and all $t^{a} \in \Gamma(T\mathcal{M})$,
\begin{eqnarray}
t(f) = t^{a}D_{a}f.
\end{eqnarray}

\item Torsion free: For all $f \in C^{r}(\cal M)$ ,
\begin{eqnarray}
D_{a}D_{b} f = D_{b}D_{a}f.
\end{eqnarray}
\end{enumerate}
\end{definition}

\begin{lemma}\label{derivativelemma}
Consider a manifold $\cal M$ endowed with two metric tensors $g$ and $\hat{g}$ and their respective derivative operators $\nabla_{a}$ and $\hat{\nabla}_{a}$ (i.e., $\nabla_{c} g_{ab}=0$ and $\hat{\nabla}_{c} \hat{g}_{ab}=0$). Then, for any tensor field $T^{b_{1}\dots b_{k}}_{\ \ \ \ \ \ \ c_{1}\dots c_{l}}$ we have
\begin{eqnarray}
\hat{\nabla}_{a}T^{b_{1}\dots b_{k}}_{\ \ \ \ \ \ \ c_{1}\dots c_{l}} = \nabla_{a} T^{b_{1}\dots b_{k}}_{\ \ \ \ \ \ \ c_{1}\dots c_{l}} &+& \sum_{i} C^{b_{i}}_{\ \ ad} T^{b_{1}\dots d\dots b_{k}}_{\ \ \ \ \ \ \ \ \ \ c_{1}\dots c_{l}}\nonumber\\
&-&\sum_{j}C^{d}_{\ \ ac_{j}} T^{b_{1}\dots b_{k}}_{\ \ \ \ \ \ \ c_{1}\dots d\dots c_{l}},
\end{eqnarray}
where
\begin{eqnarray}\label{Csymbols}
C^{c}_{\ ab} = \frac{1}{2}\hat{g}^{cd} \left\{ \nabla_{a}\hat{g}_{bd}+\nabla_{b}\hat{g}_{ad}-\nabla_{d}\hat{g}_{ab}\right\}.
\end{eqnarray}
\end{lemma}
\begin{proof}
See \cite{wald84}, chapter 3.
\end{proof}

The reader should reckon the expression for $C^{c}_{\ ab}$ as a generalized Christoffel symbol. In fact, the Christoffel symbols are obtained using the same procedure and choosing $\nabla_{a}$ and $\partial_{a}$ as derivative operators. Since it depends on the coordinate system used to define the derivative operator $\partial_{a}$, the Christoffel symbols are not a true tensor in another coordinate system. Using lemma \ref{derivativelemma} we have, for any $1-$form field $\omega_{b}$, that
\begin{eqnarray}
\left(\hat{\nabla}_{a} - \nabla_{a}\right)\omega_{b} &=& \hat{\nabla}_{a}\omega_{b} - \nabla_{a}\omega_{b}\nonumber\\
&=& \partial_{a}\omega_{b} - \hat{\Gamma}_{ab}^{c}\omega_{c} - \left(\partial_{a}\omega_{b} - \Gamma_{ab}^{c}\omega_{c}\right)\nonumber\\
&=& \left(\Gamma_{ab}^{c}-\hat{\Gamma}_{ab}^{c}\right)\omega_{c} .
\end{eqnarray}
On the other hand, using lemma \ref{derivativelemma}, we have $\hat{\nabla}_{a}\omega_{b} = \nabla_{a} \omega_{b} - C_{\ ab}^{c}\omega_{c}$, yielding
\begin{equation}\label{gammarelation}
\hat{\Gamma}_{ab}^{c} = \Gamma_{ab}^{c} +C_{\ ab}^{c}.
\end{equation}
Therefore, once we know $\Gamma_{ab}^{c}$, the knowledge of $\hat{\Gamma}_{ab}^{c}$ determines $C_{\ ab}^{c}$, and vice-versa. Let us see how to relate the curvature tensors associated with the two different affine connections $\nabla$ and $\hat{\nabla}$ in the form of the following
\begin{proposition}
The curvature tensor $\hat{R}_{abc}^{\ \ \ \ d}$ associated with the metric $\hat{g}$ in terms of the geometry defined by $g$ is given by
\begin{eqnarray}\label{riemanncurvature}
\hat{R}_{abc}^{\ \ \ \ d} = R_{abc}^{\ \ \ \ d} - 2\nabla_{[a}C^{d}_{\ b]c}+ 2C^{e}_{\ c[a}C^{d}_{\ b]e}.
\end{eqnarray}
\end{proposition}
\begin{proof}
By definition
\begin{eqnarray}
 \hat{R}_{abc}^{\ \ \ \ d}\omega_{d}  &=& \left[ \hat{\nabla}_{a},\hat{\nabla}_{b}\right]\omega_{c}\nonumber\\
&=& \hat{\nabla}_{a}(\hat{\nabla}_{b}\omega_{c}) -\hat{\nabla}_{b}(\hat{\nabla}_{a}\omega_{c})\nonumber\\
&=& \nabla_{a}(\hat{\nabla}_{b}\omega_{c}) - C^{d}_{\ ab}(\hat{\nabla}_{d}\omega_{c})-C^{d}_{\ ac}(\hat{\nabla}_{b}\omega_{d}) - (a \leftrightarrow b),
\end{eqnarray}
where $a \leftrightarrow b$ represents the same expression but interchanging the indices $a$ and $b$. Replacing $\hat{\nabla}_{m}\omega_{n} = \nabla_{m}\omega_{n} - C_{\ mn}^{p}\omega_{p}$ and performing some index substitutions one gets
\begin{eqnarray}
\left[ \hat{\nabla}_{a},\hat{\nabla}_{b}\right]\omega_{c} = R_{abc}^{\ \ \ \ d}\omega_{d} + \left[\nabla_{b}C^{d}_{\ ac} - \nabla_{a}C^{d}_{\ bc} + C^{e}_{\ ac}C^{d}_{\ be} - C^{e}_{\ bc}C^{d}_{\ ae}\right] \omega_{d},
\end{eqnarray}
therefore concluding the proof.
\end{proof}
Note that when $\hat{\nabla} = \nabla$ the $C$ symbols are all zero and the geometry is kept the same, as expected.
\begin{corollary}
The Ricci and scalar curvature associated with $\hat{g}$ are, respectively, given by
\begin{eqnarray}
\hat{R}_{ac}&=& R_{ac}- 2\nabla_{[a}C^{b}_{\ b]c}+ 2C^{e}_{\ c[a}C^{b}_{\ b]e} \\
\hat{R} &=& \hat{g}^{ac} \hat{R}_{ac}.
\end{eqnarray}
\end{corollary}
\subsection{The geometry of conformal transformations}\label{conformalsubsec}
Let $(\mathcal{M},g_{ab})$ be a spacetime. A conformal transformation of $(\mathcal{M},g_{ab})$, denoted by $(\mathcal{M},\tilde{g}_{ab})$,  is essentially a local angle-preserving change of scale where the new metric tensor is given by
\begin{eqnarray}
\tilde{g}_{ab} &=& \Omega^{2} g_{ab}.
\end{eqnarray}

It is worthwhile to mention that a conformal transformation, as defined here, {\it is not} a change of coordinates, but an actual change of the geometry. Formally, it should be written $(\mathcal{M},g_{ab}) \longmapsto (\tilde{\mathcal{M}},\tilde{g}_{ab})$. However, the spacetime $(\tilde{\mathcal{M}},\tilde{g}_{ab})$ is a subset of the manifold $\mathcal{M}$ endowed with another metric tensor defined on it, hence the abuse of notation. It should be clear that conformal transformations are not, in general, associated with a diffeomorphism of $\mathcal{M}$ \cite{wald84}.

Because now one can consider two metric tensors defined on $\mathcal{M}$, hence two affine connections, it is of major interest to use conformal transformations to change our dynamical variables: anything that is a function of $g_{ab}$ can be equally thought as a function of $\tilde{g}_{ab}$ and $\Omega$. We say that these quantities are expressed in the {\it conformal frame}.  To resume this section we list some quantities of interest in the conformal frame.
\begin{proposition}\label{conformalvariables}
In the conformal frame we have:
\begin{eqnarray}
C^{c}_{\ ab} &=& 2\delta_{\ (a}^{c} \nabla_{b)}\ln\Omega - g_{ab}g^{cd}\nabla_{d}\ln\Omega\\
\tilde{R}_{abc}^{\ \ \ \ d}&=&R_{abc}^{\ \ \ \ d} + 2\delta^{d}_{\ [a}\nabla_{b]}\nabla_{c}\ln\Omega - 2g^{de}g_{c[a}\nabla_{b]}\nabla_{e}\ln\Omega + 2(\nabla_{[a}\ln\Omega)\delta^{d}_{\ b]}\nabla_{c}\ln\Omega \nonumber\\
&-& 2(\nabla_{[a}\ln\Omega)g_{b]c}g^{df}\nabla_{f}\ln\Omega - 2g_{c[a}\delta^{d}_{\ b]}g^{ef}(\nabla_{e}\ln\Omega)(\nabla_{f}\ln\Omega)\\
\tilde{R}_{ac}&=&R_{ac}-(n-2)\nabla_{a}\nabla_{c}\ln\Omega -g_{ac}g^{de}\nabla_{d}\nabla_{e}\ln\Omega\nonumber\\
&+&(n-2)(\nabla_{a}\ln\Omega)(\nabla_{c}\ln\Omega)-(n-2)g_{ac}g^{de}(\nabla_{d}\ln\Omega)(\nabla_{e}\ln\Omega)\\
\tilde{R}&=&\Omega^{-2}\Big\{ R- 2(n-1)g^{ac}\nabla_{a}\nabla_{c}\ln\Omega -(n-2)(n-1)g^{ac}(\nabla_{a}\ln\Omega)(\nabla_{c}\ln\Omega)\Big\} \\
\tilde{\square}\phi&=& \Omega^{-2}\square \phi+(n-2)g^{ab}\Omega^{-3}(\nabla_{a}\Omega)(\nabla_{b}\phi),
\end{eqnarray}
 for any $C^{r}$($r\geq 2$) scalar field $\phi$.
\end{proposition}
\begin{proof}
Check the appendix D of \cite{wald84}.
\end{proof}


\section{The geometry of null-like disformal transformations}\label{disformasse}
\subsection{Connection, curvature and the d'Alembertian}\label{disfgeom}
\par
Disformal transformations can be seen as a generalization of conformal transformations. As such, they do not represent a change in coordinates, but a local change in the geometry instead. One might think of a conformal transformation as a smooth, {\it isotropic} and infinitesimal stretch at a point, whereas a disformal transformation is a smooth, {\it anisotropic} and infinitesimal stretch at a point.
Given a spacetime $(\mathcal{M},g_{ab})$, a null-like vector $V^{c}$ and two spacetime-dependent scalars $\alpha$ and $\beta$ with $\alpha>0$, we define a null-like disformal transformation $(\mathcal{M},g_{ab},V^{c},\alpha,\beta)\longmapsto (\mathcal{M},\hat{g}_{ab})$ as a change in geometry when the metric tensor changes according to
\begin{eqnarray}
\hat{g}_{ab} = \alpha g_{ab} + \beta V_{a}V_{b}\equiv\alpha g_{ab} + \beta g_{ac}g_{bd} V^{c}V^{d}.
\end{eqnarray}
It is easy to check that the inverse of the disformal metric in that case is given by
\begin{eqnarray}
\hat{g}^{ab} = \frac{1}{\alpha} g^{ab} -\frac{\beta}{\alpha^{2}} V^{a}V^{b}.
\end{eqnarray}
Since we are now dealing with a manifold endowed with two metric tensors, it is important to distinguish which metric tensor is being used when raising and lowering indices. One shall deal with this problem by explicitly writing the metric in all formulae in which indices are raised or lowered.

We can now consider some of the dynamical variables in the {\it disformal frame}. Using the definition of the tensor $C^{c}_{\ ab}$ in lemma \ref{derivativelemma} we find
\begin{eqnarray}
C^{c}_{\ ab} &=& \frac{1}{2\alpha}\Big[  2\alpha_{(a}\delta^{c}_{b)} + 2\beta V^cV_{(a;b)} + V^c\beta_{(a}V_{b)} +2\beta V^{c}_{\ (;a}V_{b)} + V^c\beta_{(a}V_{b)}\nonumber\\
&-& g_{ab}g^{cd}\alpha_{d} - 2\beta V_{(a;|d|}g^{cd}V_{b)}-\beta_dg^{cd}V_{(a}V_{b)}\Big]\nonumber\\ &-& \frac{1}{2\alpha^{2}}\Big[2 \beta\alpha_{(a}V_{b)}V^{c} - g_{ab}V^{c}\dot{\alpha}\beta-2\beta^2\dot{V}_{(a}V_{b)}V^{c}-\beta\dot{\beta}V_{(a}V_{b)}V^c\Big],
\end{eqnarray}

\begin{eqnarray}\label{scalarcurvature}
&\hat{R}&\, =\, \frac{1}{\alpha}R - \frac{\beta}{\alpha^{2}}R_{ac}V^{a}V^{c}+ \frac{1}{\alpha^2}\left\{(1-n)\square\alpha +\beta V^aV^b_{\ \ ;a;b}- 2\beta V^a\square V_{a}+\beta (\nabla\cdot V)^2\right. \nonumber\\
 &+& \beta V^a_{\ \ ;b}V^b_{\ \ ;a}-2\beta g^{bc}V^a_{\ ;b}V_{a;c}+2\dot{\beta} (\nabla\cdot V)+\beta_{a;b}V^aV^b+2\beta_a\dot{V}^a+\beta\left(\nabla\cdot V\right)^{\bullet}\Big\}\nonumber\\ 
 &+&\frac{1}{\alpha^3}\Big\{-(n-1)(n-6)\alpha\star\alpha+(n-2)\beta V^aV^b\alpha_{a;b}+(n-3)\dot{\alpha}\dot{\beta}+(n-3)\beta\dot{V}^a\alpha_{a}\nonumber\\
&+& (n-3)\dot{\alpha}\beta(\nabla\cdot V)+\frac{\beta^2}{2}\dot{V}^a\dot{V}_a\Big\} + \frac{1}{4\alpha^4}(n-9)(n-2)(\dot{\alpha})^2\beta,
\end{eqnarray}
and
\begin{eqnarray}
&\hat{\square}\Phi&\, =\, \frac{1}{\alpha}\square\Phi + \frac{\beta}{\alpha^{3}}\frac{(4-n)}{2}\dot{\alpha}\dot{\Phi}\\
&+&\frac{1}{\alpha^{2}}\left[ \frac{(n-1)}{2}\Phi^{a}\alpha_{a} - \dot{\beta}\dot{\Phi} - \beta\dot{\Phi}(\nabla\cdot V) - \beta\dot{V}^{c}\Phi_{c}-\beta V^{a}V^{b}\Phi_{a;b}\right],\nonumber
\end{eqnarray}
where $n=\dim\mathcal{M}$, $\beta_{a} = \partial_{a}\beta$, $\alpha_{a} = \partial_{a}\alpha$, $\square\alpha=\nabla^{a}\nabla_{a}\alpha$, $\alpha\star\alpha = g^{ab}\partial_{a}\alpha\partial_{b}\alpha$, $\dot{\alpha_{a}}=V^{b}\nabla_{b}\partial_{a}\alpha$, $V_{a;c}=\nabla_cV_a$, $\nabla\cdot V = \nabla_{a}V^{a}$, $\square V_a=g^{bc}V_{a;b;c}$, $\dot{V}^{a}=V^b\nabla_b V^a$ and $(\nabla\cdot V)^{\bullet}=(\nabla\cdot V)_{,a}V^a$. Setting $V^{a} = 0$ and $\alpha = \Omega^{2}$ one can recover the geometry in the conformal frame given in proposition \ref{conformalvariables}.  

In the examples below, we are going to refer them to a test case. They are used to show the validity of our formulae in an actual example. In order to avoid defining it in every example, the reader should have in mind the following
\begin{definition}[The test case]
The manifold $\mathcal{M}$ is $\mathbb{R}^{4}$ with Cartesian coordinates $\{x^{\mu}\} = (t,x,y,z)$ and the background metric is the Minkowski one $g_{ab}=\mbox{diag}(1,-1,-1,-1)$. The disformal parameters are $\alpha=1$ and $\beta = \frac{2m}{r}$, and the null-like vector $V^{\mu} = \left(1,-\frac{x}{r},-\frac{y}{r}, -\frac{z}{r} \right)$, and $r=\sqrt{x^{2} + y^{2}+z^{2}}$. This is the Schwarzschild metric in the Kerr-Schild formulation. Of course one could use the Kerr-Newman metric as a test case, but this would complicate significantly the calculations without adding anything new to the problem.
\end{definition}

\begin{example}[The test case -- I] 
The derivative operator compatible with $g$ is the coordinate system derivative operator $\partial_{a}$ and $\Gamma_{ab}^{c} \equiv 0$ (and therefore The Riemann and Ricci tensors are null everywhere). Using Eq. (\ref{gammarelation}) we have
\begin{equation}
C_{\ ab}^{c} = \hat{\Gamma}_{ab}^{c},
\end{equation}
which is simply the Eq. (\ref{Csymbols}) when one replaces $\nabla_{a}$ by $\partial_{a}$, as expected. Using Eq. (\ref{riemanncurvature}) and using the fact that $R_{abc}^{\ \ \ \ d} =0$, $\nabla_{a} = \partial_{a}$ and $C_{\ ab}^{c} = \hat{\Gamma}_{ab}^{c}$ we have
\begin{eqnarray}
\hat{R}_{abc}^{\ \ \ \ d} &=& R_{abc}^{\ \ \ \ d} + \nabla_{b}C^{d}_{\ ac} - \nabla_{a}C^{d}_{\ bc} + C^{e}_{\ ac}C^{d}_{\ be} - C^{e}_{\ bc}C^{d}_{\ ae}\nonumber\\
&=&\partial_{b}\hat{\Gamma}^{d}_{ac} - \partial_{a}\hat{\Gamma}^{d}_{bc} + \hat{\Gamma}^{e}_{ac}\hat{\Gamma}^{d}_{be} - \hat{\Gamma}^{e}_{bc}\hat{\Gamma}^{d}_{ae},
\end{eqnarray}
which is the expression for the Riemann curvature tensor associated with $\hat{g}$ as expected. Substituting $\alpha=1$, $\beta=2m/r$, and $n=4$ in the expression for the scalar curvature (\ref{scalarcurvature}) we obtain $\hat{R} = 0$. 
\end{example}

\subsection{Mutual geodesics}\label{geodesicssubsec}
One might be interested in comparing the geodesics with respect to $\nabla$ with those with respect to $\hat{\nabla}$. As is the case for conformal transformations, one cannot expect the geodesics to be preserved. Nonetheless, for conformal transformations null geodesics are indeed preserved, although the geodesic in the conformal frame is not affinely parameterized. This is an extraordinary feature that allows one to study causality up to a conformal transformation (for instance, using the Carter-Penrose conformal diagrams). Since disformal transformations provide a generalization of the conformal ones, we cannot expect the geodesics to be preserved, not even in the null case. Surprisingly, if the tangent vector of a null geodesic satisfies an extra condition, this null geodesic is preserved. This result can be stated as the following
\begin{proposition}
Let $(\mathcal{M},g_{ab},V^{c},\alpha,\beta)\longmapsto (\mathcal{M},\hat{g}_{ab})$ be a null-like disformal transformation. Let $w^{a}\in\Gamma(T\mathcal{M})$ be the tangent vector to an affinely parameterized null geodesic $\gamma$ of the background metric. If $w_{a}V^{a} = w^{a}V_{a} = g_{ab}w^{a}V^{b}=0$, then $\gamma$ is a null geodesic with respect to $\hat{\nabla}$.
\end{proposition}
\begin{proof}
We have
\begin{eqnarray*}
w^{a}\hat{\nabla}_{a}w^{b} &=& w^{a}{\nabla}_{a}w^{b}+ C^{b}_{\ ac}w^{a}w^{c}\\&=&C^{b}_{\ ac}w^{a}w^{c},
\end{eqnarray*}
since $w^{a}{\nabla}_{a}w^{b}=0$. Using the expression for $C^{b}_{\ ac}$ yields
\begin{eqnarray*}
w^{a}\hat{\nabla}_{a}w^{b}&=&\Big\{\frac{1}{2\alpha}\Big[  2\alpha_{(a}\delta^{b}_{c)} + 2\beta V^bV_{(a;c)} + V^b\beta_{(a}V_{c)} +2\beta V^{b}_{\ (;a}V_{c)} + V^b\beta_{(a}V_{c)}\nonumber\\
&-& g_{ac}g^{bd}\alpha_{d} - 2\beta V_{(a;|d|}g^{bd}V_{c)}-\beta_dg^{bd}V_{(a}V_{c)}\Big]\nonumber\\ &-& \frac{1}{\alpha^{2}}\Big[2 \beta\alpha_{(a}V_{c)}V^{b} - g_{ac}V^{b}\dot{\alpha}\beta-2\beta^2\dot{V}_{(a}V_{c)}V^{b}-\beta\dot{\beta}V_{(a}V_{c)}V^b\Big]\Big\}w^{a}w^{c}\\
&=&\left(w^{a}\nabla_{a}\ln\alpha\right)w^{b}.
\end{eqnarray*}
%
%
%
%
%
The parameters of the geodesics are then related by $\frac{d\hat{\lambda}}{d\lambda}=c\alpha$, where $c\in\mathbb{R}$ is a constant, showing that it is not affinely parameterized.  
We used that $2\beta V_{(c;a)}V^{b}w^{a}w^{c} = 0$. Indeed:
\begin{eqnarray*}
2\beta V_{(c;a)}V^{b}w^{a}w^{c} &=&\left\{\beta V^{b} \nabla_{a}V_{c} + \beta V^{b} \nabla_{c}V_{a}\right\}w^{a}w^{c}\\
&=&\beta V^{b}w^{a}w^{c}\nabla_{a}V_{c}+\beta V^{b}w^{a}w^{c}\nabla_{c}V_{a}.
\end{eqnarray*}
Investigating the first term (the second term follows by analogy), we have
\begin{eqnarray*}
\beta V^{b}w^{a}\nabla_{a}(w^{c}V_{c}) = 0 = \underbrace{\beta V^{b}w^{a}w^{c}\nabla_{a}V_{c}}_{\text{first term}} + \beta V^{b}V_{c}\underbrace{w^{a}\nabla_{a}w^{c}}_{=0\ (\text{geodesic})}.
\end{eqnarray*}
\end{proof}
The proposition above shows that for the case of a null-like disformal transformation the null vector $w^{a}$ must satisfy an extra condition, to wit $V_{a}w^{a}=0$. This extra condition involves a coupling between $w^{a}$, the  disformal portion of the transformation. Therefore, since conformal transformation preserve causal relations, the study of the causal structure must rely on $\beta$ and $V^{c}$. As shown in \cite{Carvalho:2015omv}, $\beta$ plays the crucial role in the change of the causal structure. 
\begin{example}[The test case -- II]
The vector defining this disformal transformation, $V^{\mu} = \left(1,-\frac{x}{r},-\frac{y}{r},-\frac{z}{r}\right)$, satisfy the conditions of the proposition above. Therefore:
\begin{eqnarray}
V^{\mu}\nabla_{\mu} V^{\nu} = V^{\mu}\hat{\nabla}_{\mu} V^{\nu}=0. 
\end{eqnarray}
These are precisely the radial null-like geodesic of Minkowski and Schwarzschild spacetimes. 
\end{example}


\subsection{Mutual Killing vector fields}\label{killingsubsec}
Let $(\mathcal{M},g_{ab})$ be a spacetime. We say that $\xi^{a}\in\Gamma(T\mathcal{M})$ is a Killing vector field of $g_{ab}$ if
\begin{eqnarray}\label{Lie}
\mathcal{L}_{\xi}g_{ab} \equiv \xi^{c}\partial_{c}g_{ab} + (\partial_{a}\xi^{c})g_{cb}+(\partial_{b}\xi^{c})g_{ca}=0,
\end{eqnarray}
i.e., the Lie derivative of $g_{ab}$ in the direction of $\xi^{c}$ is zero. The Killing vectors of a metric are associated with the symmetries of that metric. Despite the appearance, equation (\ref{Lie}) is known to be covariant.

Needless to say, knowing the Killing vectors is of major importance. The proposition below relates the Killing vectors of the background and disformal metrics. The corollary provides a well-known result when the disformal transformation is purely conformal.

\begin{proposition}
Let $(\mathcal{M},g_{ab},V^{c},\alpha,\beta)\longmapsto (\mathcal{M},\hat{g}_{ab})$ be a null-like disformal transformation and let $\xi^{c}$ be a Killing vector field of $\hat{g}_{ab}$. Then, $\xi^{c}$ is a Killing vector field of $g_{ab}$ if, and only if, 
\begin{eqnarray}
(\xi^{c}\partial_{c}\alpha) g_{ab} + \xi^{c}\partial_{c}(\beta V_{a}V_{b})+2\beta V_{(a}\left(\partial_{b)}\xi^{c}\right)V_{c}=0.
\end{eqnarray}
\end{proposition}
\begin{proof}
Using the equation (\ref{Lie}) for the metric $\hat{g}_{ab}$, we have:
\begin{eqnarray}
\mathcal{L}_{\xi}\hat{g}_{ab} = \xi^{c}\partial_{c}\hat{g}_{ab} + (\partial_{a}\xi^{c})\hat{g}_{cb}+(\partial_{b}\xi^{c})\hat{g}_{ca}=0.
\end{eqnarray}
Replacing the expression for $\hat{g}_{ab}$ yields
\begin{eqnarray}
0=\mathcal{L}_{\xi}\hat{g}_{ab}=(\xi^{c}\partial_{c}\alpha)g_{ab} + \alpha\mathcal{L}_{\xi}g_{ab} + \xi^{c}\partial_{c}(\beta V_{a}V_{b})+2\beta V_{(a}\left(\partial_{b)}\xi^{c}\right)V_{c}.
\end{eqnarray}
Since $\alpha\neq 0$, $\mathcal{L}_{\xi}g_{ab}=0$ if, and only if, 
\begin{eqnarray}
(\xi^{c}\partial_{c}\alpha) g_{ab} + \xi^{c}\partial_{c}(\beta V_{a}V_{b})+2\beta V_{(a}\left(\partial_{b)}\xi^{c}\right)V_{c}=0.
\end{eqnarray}
\end{proof}
\begin{corollary}
For conformal transformations, $\xi^{a}$ is at least a conformal Killing vector field of $g_{ab}$ and will be a true Killing vector field if, and only if, $\alpha$ is constant along the orbits (integral curves) of $\xi^{a}$.
\end{corollary}
\begin{example}[The test case -- III]
In that case, $\alpha\equiv1$ and therefore the necessary and sufficient condition for $\xi^{c}$ to be a mutual Killing vector field of the Minkowski and Schwarzschild metrics is
\begin{eqnarray}\label{killingtestcase}
\xi^{c}\partial_{c}(\beta V_{a}V_{b})+2\beta V_{(a}\left(\partial_{b)}\xi^{c}\right)V_{c}=0.
\end{eqnarray}
The rotational Killing vectors $R= -y\partial_{x}+x\partial_{y}$, $S=z\partial_{x}-x\partial_{z}$, $T=-z\partial_{y}+y\partial_{z}$ and the time-translation Killing vector $\partial_{t}$ are shown to satisfy the equation (\ref{killingtestcase}). Therefore, they are Killing vector fields for both metrics.
\end{example}

\subsection{Disformal Killing equation}\label{disformalkillingeq}
In this section we generalize the notion of conformal Killing vectors to the disformal case. It is shown that, under some hypotheses, it is possible to find a solution for this disformal Killing equation. The meaning of this solution is discussed in the end of this section. We begin with the following
\begin{definition}
A vector field $X^c$, satisfying the equation
\begin{equation}\label{disfkilldef}
\left({\cal L}_Xg\right)_{ab}=\alpha\, g_{ab} + \beta\, V_aV_b,
\end{equation}
is called a {\it null-like disformal Killing vector} of the metric $g$, where $\alpha$ and $\beta$ are scalar fields and $V^c$ is a null-like vector field.
\end{definition}
It would be an exercise in futility to define such an object if no solution to equation (\ref{disfkilldef}) existed. Let us consider a flat metric $\eta_{ab}$ and the $1-$form field $\sqrt{\beta}V_{a}=U_a=\partial_a\phi$ for scalar field $\phi$. We analyze this particular case following the notation of \cite{choquet-bruhat}.
\par
From the trace of the equation (\ref{disfkilldef}), the null-like disformal Killing equation is
\begin{equation}\label{k0}
\partial_aX_b+\partial_bX_a=\frac{2}{n}\eta_{ab}\partial_cX^c+U_aU_b
\end{equation}
Defining $\psi=\frac{2}{n}\partial_cX^c$ and differentiating equation (\ref{k0}) we have
\begin{eqnarray}
\partial_c\partial_aX_b+\partial_c\partial_bX_a=\eta_{ab}\partial_c\psi+U_b\partial_cU_a+U_a\partial_cU_b,\label{k1}\\
\partial_a\partial_bX_c+\partial_a\partial_cX_b=\eta_{bc}\partial_a\psi+U_c\partial_aU_b+U_b\partial_aU_c.\label{k2}
\end{eqnarray}
Subtracting (\ref{k1}) and (\ref{k2}):
\begin{equation}\label{k3}
\partial_b\left(\partial_cX_a-\partial_aX_c\right)=\eta_{ab}\partial_c\psi-\eta_{bc}\partial_c\psi+U_a\partial_cU_b-U_c\partial_aU_b.
\end{equation}
Differentiating again, considering that $\partial_d\partial_b\left(\partial_cX_a-\partial_aX_c\right)=\partial_d\partial_b\left(\partial_cX_a-\partial_aX_c\right)$, and contracting the result with $\eta^{ad}\eta^{bc}$, yields
\begin{equation}\label{b1}
(n-1)\Box \psi+\eta^{ad}\eta^{bc}\left(\partial_b\partial_a\phi\right)\left(\partial_c\partial_d\phi\right)-\left(\Box\phi\right)^2=0,
\end{equation}
where $\Box=\eta^{ab}\partial_a\partial_b$. For simplicity, let us assume the case in which $\partial_{\alpha}\phi=2C_{\alpha}$ (constant) in Cartesian coordinates (now we are using greek indices because we are fixing the coordinate system to be  the Cartesian one). This way, Eq. (\ref{b1}) is simply
\begin{equation}
\Box\psi=0.
\end{equation}
We will chose the solution to be
\begin{eqnarray}
\psi=2B\, +\, 4B_{\alpha}x^{\alpha},
\end{eqnarray}
where $B, B_{\alpha}$ are constants. Substituting in (\ref{k3}), and integrating
\begin{equation}\label{k4}
\partial_{\lambda}X_{\alpha}-\partial_{\alpha}X_{\lambda}=4\left(B_{\lambda}x_{\alpha}-B_{\alpha}x_{\lambda}\right)+2A_{\alpha\lambda},
\end{equation}
where $A_{\alpha\lambda}=-A_{\lambda\alpha}$. Summing Eqs. (\ref{k0}) and (\ref{k4}), we have
\begin{equation}
\partial_{\lambda}X_{\alpha}=2(B_{\lambda}x_{\alpha}-B_{\alpha}x_{\lambda})+A_{\alpha\lambda}+\eta_{\alpha\lambda}\left(B+2B_{\beta}x^{\beta}\right)+2C_{\lambda}C_{\alpha}.
\end{equation}
Finally, integrating this equation we find
\begin{equation}
X_{\alpha}=A_{\alpha}+(A_{\alpha\lambda}+2C_{\alpha}C_{\lambda})x^{\lambda}+Bx_{\alpha}+2B_{\lambda}x_{\alpha}x^{\lambda}-B_{\alpha}x_{\lambda}x^{\lambda},
\end{equation}
where $A_{\alpha}$ are arbitrary constants.
\par
This solution can be decomposed as a linear combination of the generators of the conformal group \cite{choquet-bruhat}
\begin{eqnarray*}
{\bf p}_{\alpha}=\partial_{\alpha}, \ \ \ {\bf m}_{\alpha\beta}=x_{\alpha}\partial_{\beta}-x_{\beta}\partial_{\alpha},\\
{\bf d}=x^{\alpha}\partial_{\alpha}, \ \ \ {\bf k}_{\alpha}=2x_{\alpha}x^{\beta}\partial_{\beta}-x_{\beta}x^{\beta}\partial_{\alpha},
\end{eqnarray*}
and the fixed vector
\begin{equation}
{\bf w}=2C^{\alpha}C_{\beta}x^{\beta}\partial_{\alpha}.
\end{equation}
\par
To illustrate our construction, consider the coordinate transformation induced by ${\bf w}$:
\begin{equation}
x'^{\alpha}=x^{\alpha}+2C^{\alpha}C_{\beta}x^{\beta}.
\end{equation}
This way, the light cone transforms as
\begin{equation}
\eta_{\alpha\beta}x'^{\alpha}x'^{\beta}=\left(\eta_{\alpha\beta}+4C_{\alpha}C_{\beta}\right)x^{\alpha}x^{\beta}=\left(\eta_{\alpha\beta}+U_{\alpha}U_{\beta}\right)x^{\alpha}x^{\beta}.
\end{equation}
Which means that a coordinate transformation induced by the disformal Killing equation preserves the disformal structure.


\subsection{Generalized Weyl transformations}\label{weylsubsec}
Weyl geometry is a generalization of Riemannian geometry, that presents a non-metricity tensor, i.e., the compatibility between the metric of the manifold and the connection is determined by the rule
\begin{equation}\label{w1}
\nabla_cg_{ab}=\sigma_cg_{ab},
\end{equation}
where $\sigma_c$ is a $1$-form field named the Weyl field. By assuming this {\it W-compatibility condition}, along with the torsionless connection requirement, it is possible to generalize Levi-Civita's theorem to find the unique connection that satisfies (\ref{w1}), given by
\begin{equation}\label{w2}
\Gamma^c_{ab}=\frac{1}{2}g^{cd}\left(\partial_a g_{bd}+\partial_b g_{ad}-\partial_d g_{ab}\right)-\frac{1}{2}g^{cd}\left(\sigma_a g_{bd}+\sigma_b g_{ad}-\sigma_d g_{ab}\right).
\end{equation}
\par
Such geometry presents an inherent symmetry property. By performing the simultaneous transformation, named Weyl transformations, $(g_{ab}\, ,\, \sigma_c)\mapsto (e^fg_{ab}\, ,\, \sigma_c+f_{,c})$, the form of the compatibility condition (\ref{w1}) and the connection coefficients (\ref{w2}) are preserved (here $f=f(x)$ is a scalar function). Which means that for the simultaneous transformations $\widehat{g}_{ab}\doteq e^fg_{ab}$ and $\widehat{\sigma}_c\doteq \sigma_c+f_{,c}$, the compatibility condition is transformed to
\begin{equation}
\nabla_c\widehat{g}_{ab}=\widehat{\sigma}_c\widehat{g}_{ab},
\end{equation}
and consequently the connection coefficients have the same functional dependence as (\ref{w2}), but depending of the pair $(\widehat{g}_{ab},\widehat{\sigma}_c)$. This means that different choices of the scalar function $f$ defines different {\it frames} in this geometry.
\par
An interesting possibility consists in defining $\sigma=d\phi$, i.e., the Weyl field as the differential of a scalar field $\phi$. This defines the so called integrable Weyl geometry. This way, looking at the Weyl transformations, by choosing $f=-\phi$, there exists a Riemannian frame 
\begin{equation}
\left(g_{ab},\phi\right)\longmapsto \left(\widehat{g}_{ab}=e^{-\phi}g_{ab},d\widehat{\phi}=0\right),
\end{equation}
meaning that the integrable Weyl geometry can be effectivelly described as a Riemannian geometry with an effective conformal metric. In fact, considering the conformal relation $\widehat{g}_{ab}=e^{-\phi}g_{ab}$, it is straightforward to check that it is invariant under the Weyl transformation, i.e., if we transform $\tilde{g}_{ab}=e^fg_{ab}$ and $\tilde{\phi}=\phi+f$, we have
\begin{equation}\label{w4}
\widehat{g}_{ab}=e^{-\phi}g_{ab}\longmapsto e^{-\tilde{\phi}}\tilde{g}_{ab}.
\end{equation}
\par
Since there exists a Riemannian frame, a different form of treating the integrable geometry and deriving the Weyl transformation consists in treating it as a Riemannian geometry in an effective, conformal metric. This way, the Riemannian compatibility condition
\begin{equation}
\nabla_c\left(e^{-\phi}g_{ab}\right)=0
\end{equation}
is equivalent to equation (\ref{w1}) for $\sigma=d\phi$. And even the connection coefficients (\ref{w2}) are the Christoffel symbols of $e^{-\phi}g_{ab}$. Since the Weyl transformation preserves the form of the effective metric by $(\ref{w4})$, it also preserves the W-compatibility condition induced by the conformal metric
\begin{equation}
\nabla_c\left(e^{-\phi}g_{ab}\right)=0\Leftrightarrow \nabla_cg_{ab}=\phi_{,c}g_{ab}.
\end{equation}
\par
This construction can be generalized to an induced geometry inspired by a disformal transformation (see \cite{Yuan:2015tta} for the time-like case). Consider a disformal relation 
\begin{equation}\label{d1}
\widehat{g}_{ab}=\alpha g_{ab}+\beta g_{ac}g_{bd}V^{c}V^{d}.
\end{equation}
Performing a disformal transformation $(g_{ab},V^{c},\lambda,\gamma)\longmapsto (\tilde{g}_{ab})$,
\begin{equation}
\tilde{g}_{ab}=\lambda g_{ab}+\gamma g_{ac}g_{bd}V^cV^d.
\end{equation}
Inverting this relation we have\footnote{Unlike the time-like case, in the null case, the inverse transformation is not equal to the transformation of the inverse metric.}
\begin{equation}
g_{ab}=\frac{1}{\lambda}\tilde{g}_{ab}-\frac{\gamma}{\lambda^3}\tilde{V}_a\tilde{V}_b,
\end{equation}
where $\tilde{V}_a\doteq \tilde{g}_{ab}V^b$.
Substituting in Eq. (\ref{d1}), we have
\begin{equation}
\widehat{g}_{ab}=\frac{\alpha}{\lambda}\tilde{g}_{ab}+\frac{1}{\lambda^2}\left(\beta-\frac{\alpha\, \gamma}{\lambda}\right)\tilde{V}_a\tilde{V}_b.
\end{equation}
If we define
\begin{equation}
\tilde{\alpha}=\frac{\alpha}{\lambda}, \ \ \ \tilde{\beta}=\frac{1}{\lambda^2}\left(\beta-\frac{\alpha\, \gamma}{\lambda}\right),
\end{equation}
the disformal relation is preserved. Therefore, for the simultaneous transformations
\begin{eqnarray}\label{t1}
\left\{
\begin{array}{ccc}
\tilde{g}_{ab}&=&\lambda\, g_{ab}+\gamma V_aV_b,\\
\tilde{V}_a&=&\lambda V_a,\\
\tilde{\alpha}&=&\alpha/\lambda,\\
\tilde{\beta}&=&(\beta-\alpha\gamma/\lambda)/\lambda^2,
\end{array}
\right.
\end{eqnarray}
we shall have 
\begin{equation}
\alpha\, g_{ab}+\beta\, V_aV_b\longmapsto \tilde{\alpha}\, \tilde{g}_{ab}+\tilde{\beta}\, \tilde{V}_a\tilde{V}_b\, .
\end{equation}
For the case of conformal transformations, as we saw above, the simultaneous transformations that preserve the conformal relation are
\begin{eqnarray}\label{t2}
\left\{
\begin{array}{ccc}
\tilde{g}_{ab}&=&\lambda\, g_{ab}\, , \\
\tilde{\alpha}&=&\alpha/\lambda\, .
\end{array}
\right.\Longrightarrow  \alpha\, g_{ab}\longmapsto \tilde{\alpha}\, \tilde{g}_{ab}.
\end{eqnarray}
These are the Weyl transformations, and the preservation of the conformal relation is related the invariance of the compatibility condition in integrable Weyl geometry.
\par
The simultaneous transformations (\ref{t1}), can be regarded as those that preserve the compatibility condition between the metric and the connection induced by a disformal transformation in the sense of \cite{Yuan:2015tta}, where the functions $(\lambda,\gamma)$ label different frames of the induced geometry.

\section{The disformal operator revisited}\label{disfopsec}
In Refs. \cite{Bittencourt:2015ypa,Carvalho:2015omv} the group structure and disformal operator were studied in detail for the time-like case. However, some algebraic differences occur when one deals with null-like disformal transformations/operators. Remember that time-like disformal transformation is defined as 
\begin{eqnarray}
\hat{g}_{ab}&=&\alpha g_{ab}+\frac{\beta}{V^2}V_{a}V_{b},\label{cov}
\end{eqnarray}
with inverse given by
\begin{eqnarray}
\hat{g}^{ab}&=&\frac{1}{\alpha}g^{ab}-\frac{\beta}{\alpha(\alpha+\beta)}\frac{V^{a}V^{b}}{V^2}.\label{contrav}
\end{eqnarray}
In \cite{Carvalho:2015omv} it was verified that the operator $\hat{g}^{a}_{\ b}\equiv\hat{g}_{cb}{g}^{ac}$  is the square of the disformal operator, and that such an operator has a basis of eigenvectors. In fact, one can choose an orthonormal basis (orthonormal with respect to $g$) at a point $p\in\mathcal{M}$ where $\sqrt{V^2}e^a_{(0)} = V^a $ to show that, in that basis,
\begin{eqnarray}
\hat{g}=
\left(\begin{array}{cccc}
  \alpha+\beta  & 0 &0 & 0 \\
  0 & -\alpha & 0&0 \\
  0 & 0 & -\alpha &0 \\
  0 & 0 & 0 &-\alpha \\
\end{array}
\right),
\end{eqnarray}
and therefore, 
\begin{eqnarray}
D^2_{\mbox{time-like}}=
\left(\begin{array}{cccc}
  \alpha+\beta  & 0 &0 & 0 \\
  0 & \alpha & 0&0 \\
  0 & 0 & \alpha &0 \\
  0 & 0 & 0 &\alpha \\
\end{array}
\right).
\end{eqnarray}
The trick consisting of taking a unit vector in the direction of $V^{a}$ and completing this set to an orthonormal basis proves that the disformal operator has a basis of eigenvectors.

Intending to extend this idea to the case of a null-like disformal transformation, the overt problem is that $V^a$ cannot belong to an orthonormal basis. One could, however, do the second best thing: Consider a null-like tetrad basis at an arbitrary point $p\in\mathcal{M}$ (i.e., a Newman-Penrose basis) where $V^a=e^a_{(0)}$. The background metric in this basis is given by
\begin{eqnarray}
g=
\left(\begin{array}{cccc}
 0  & 1 &0 & 0 \\
  1 & 0 & 0&0 \\
  0 & 0 & 0 &-1 \\
  0 & 0 & -1 &0 \\
\end{array}
\right).
\end{eqnarray}
In this case we have
\begin{eqnarray}\label{disfopnull}
\hat{g}=
\left(\begin{array}{cccc}
 \beta  & \alpha &0 & 0 \\
  \alpha & 0 & 0&0 \\
  0 & 0 & 0 &-\alpha \\
  0 & 0 & -\alpha &0 \\
\end{array}
\right) \Longrightarrow
D^2_{\mbox{null}}=
\left(\begin{array}{cccc}
  \alpha  & \beta &0 & 0 \\
  0 & \alpha & 0&0 \\
  0 & 0 & \alpha&0 \\
  0 & 0 & 0 &\alpha \\
\end{array}
\right).
\end{eqnarray}
The null tetrad basis chosen is almost a Jordan basis. Changing $e^a_{(0)}\rightarrow\beta e^a_{(0)}$ and keeping the other vectors in our basis unchanged, we have our Jordan basis and the disformal operator is given by
\begin{eqnarray}\label{finaljordan}
D^2_{\mbox{null}}=
\left(\begin{array}{cccc}
  \alpha  & 1 &0 & 0 \\
  0 & \alpha & 0&0 \\
  0 & 0 & \alpha&0 \\
  0 & 0 & 0 &\alpha \\
\end{array}
\right).
\end{eqnarray}
It is now obvious that $D^2$ has only $\alpha$ as an eigenvalue with algebraic multiplicity $4$, geometric multiplicity $3$ and  the minimal polynomial given by $p_{M}(x) = (\alpha - x)^2$ .  Hence, no basis of eigenvectors of $D^2_{\mbox{null}}$ can exist. Finally, it is easy to verify that the product of matrices of the type (\ref{disfopnull}) is a matrix of the same type, showing that the group structure is also satisfied by null type disformal operators/transformations.

Time-like disformal transformations keep the causal character of the vector used in the disformal transformation. This is also the case for null-like transformations, as the reader can easily verify (the fact to the matter is that the vector used to define disformal transformations is always an eigenvector of the disformal operator, therefore the subspace spanned by it is preserved). As a result, in the null case the background and foreground light cones coincide along one null direction (one, and only one, since the last two eigenvectors of (\ref{finaljordan}), although null-like in character, are complex), which explains the narrow relation between null-like disformal transformations and the conformal ones pointed out in this work (see also figure \ref{fig1} below). More information can be found in an independent study performed in \cite{Baccetti:2012ge}.

\begin{figure}[H]
\centering
\subfloat{\includegraphics[width=60mm]{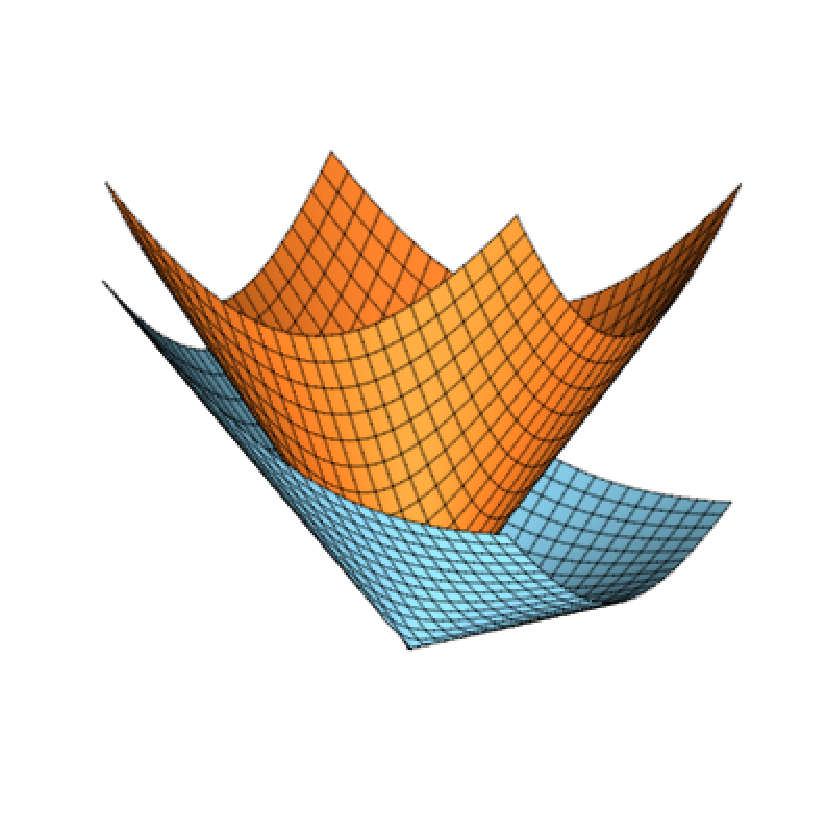}}
\hfill
\subfloat{\includegraphics[width=60mm]{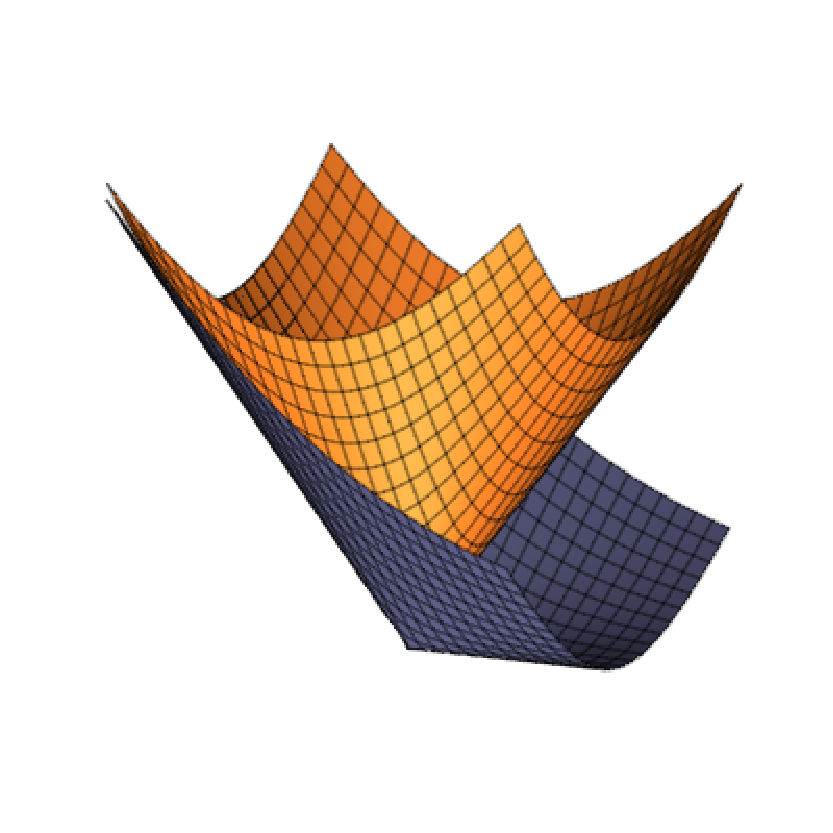}}
\caption{\footnotesize{The figure in the left represents the relation between the light cones after a time-like disformal transformation (check that the only coincide at the origin) whereas the figure on the right represents the relation between the light cones after a null-like disformal transformation (check that it is nearly conformal in the vicinity of the shared null direction).}}
\label{fig1}
\end{figure}

\section{Discussion}\label{conclusion}
We revised how some geometric objects change when two affine connections are defined on $\mathcal{M}$ and applied the results for the case of a null-like disformal transformation. As pointed out throughout the text, similar results can be found for both null-like disformal transformations and conformal transformations. Many of these similarities are due to the geometric fact that the light cones of a null-like disformal metric share a null direction with the background light cones (whereas the conformal light cones share all null directions), and are not expected to be true for time-like disformal transformations. Furthermore, the new results presented here are accompanied with important physical examples. 

To motivate further studies in this area, we introduced the concept of a disformal Killing vector field showing a explicit solution in a particular case. We also generalized the concept of Weyl transformations that preserve the null-like disformal relation.

We revisited the disformal operator and its algebraic properties, showing that it is not diagonalizable in the present case and the suitable basis to study it is a basis consisting of null vectors. 

Since we decomposed the disformal transformation of the metric into a disformal operator acting on vector fields, we wondered if it would be possible to further decompose this transformation as an action on spinors. Therefore, we generalized the disformal transformation beyond the spacetime version to the spinor space. We analysis led us to conclude that there is no disformal transformation in the spin basis that might induce a disformal metric either in the null-like or time-like case (except the conformal case). 

The present systematic study might be useful for future developments on disformal gravity, now using a null-like vector field. In such case, it has been shown in \cite{Galtsov:2018xuc} that as opposite to the time-like case, gravitational waves would propagate with the speed of light which is compatible with the recent detections by LIGO.

\appendix
\section{Spinors and disformal transformation}\label{app-disf}
As mentioned in the abstract, there is a hindrance in defining a disformal transformation in the spinor space that propagates to the spacetime geometry. This is not the  case for conformal transformations, as shown below, and can be find in standard text on the subject \cite{wald84,stewart}. This appendix was organized to be as self-contained as possible. For a comprehensive text about the subject the reader is encouraged to read the aforementioned references as well as \cite{penroserindler}.
\subsection{Spinors and spinor space}
Let us start with a two-dimensional vector space over $\mathbb{C}$. The elements in $S$ should be denoted with a superscript\footnote{We shall adopt the following index convention: capital Latin letters for spinorial objects.}. For example $\xi^A$ is an element of S.
The set of linear maps from $S$ to $\mathbb{C}$, i.e. the {\it dual space} of $S$, is going to be denoted by $S^*$. We shall denote its elements with subscripts. Thus, for example, $\eta_A$ is an element of $S^*$. The set of anti-linear maps from $S$ to $\mathbb{C}$ is denoted by $\bar{S}^*$ and it is called the {\it complex conjugate dual space}. We shall denote its elements with primed subscripts. Thus, for example, $\psi_{A^\prime}$ is an element of $\bar{S}^*$. Finally, the {\it complex conjugate space} is the dual space of $\bar{S}^*$ and it is denoted by $\bar{S}$. The elements of $\bar{S}$ should be denoted with primed superscript such as in $\phi^{A^\prime}$. All four vector spaces defined above are two-dimensional vector spaces over the complex numbers $\mathbb{C}$.

A tensor $T$ of type $(k,l;k^\prime,l^\prime)$ over $S$ is a multilinear map
\begin{eqnarray}
T: \underbrace{S^*\times\ldots\times S^*}_{k}\times\underbrace{S\times\ldots\times S}_{l}\times\underbrace{\bar{S}^*\times\ldots\times \bar{S}^*}_{k^\prime}\times\underbrace{\bar{S}\times\ldots\times \bar{S}}_{l^\prime}\longrightarrow \mathbb{C},
\end{eqnarray}
and we shall use a natural generalization of the index notation used for tensors. The relative order of primed and unprimed indices is irrelevant whereas the relative order within unprimed indices and the relative order within primed indices is relevant as in the case of tensors over a real vector space. Contraction is to be defined as usual with the only exception being that it should be realized within primed or within unprimed indices and never between one primed and one unprimed index.

A {\it symplectic form} on $S$ is a skew-symmetric and nondegenerate tensor of type $(0,2;0,0)$ over $S$. If a particular such tensor $\epsilon_{AB} = -\epsilon_{BA}$ is chosen, the pair $(S,\epsilon_{AB})$ is called {\it spinor space}, the elements of $S$ are called {\it spinors} and tensors over $S$ are called {\it spinorial tensors}. Henceforth we shall call $\epsilon_{AB}$ the $\epsilon$-spinor.

We can use $\epsilon_{AB}$ to map spinors into dual spinors via $\xi^A \mapsto \epsilon_{AB}\xi^A$, and since $\epsilon_{AB}$ is nondegenerate, we obtain an isomorphism between $S$ and $S^*$ much like the one that would be obtained using an inner product on $S$. As usual we denote $\xi_B = \epsilon_{AB}\xi^A$. We use this isomorphism to lower any unprimed index of spinorial tensors.

Note, however, that since $\epsilon_{AB}$ is skew-symmetric it makes a difference which index of $\epsilon_{AB}$ is being contracted.  Following the standard conventions, $\epsilon^{AB}$ is defined to be {\it minus} the inverse of $\epsilon_{AB}$, i.e. the skew-symmetric tensor of type $(2,0;0,0)$ which satisfies
\begin{eqnarray}
\epsilon^{AB}\epsilon_{BC} = -\delta^A_{\ \ C}.
\end{eqnarray}
Furthermore, we should use $\epsilon^{AB}$ to raise unprimed indices. To prevent undesired negative signs one should just pay attention to the rules
\begin{eqnarray}
\xi_B = \epsilon_{AB}\xi^A = - \epsilon_{BA}\xi^A,\ \ \mu^A = \epsilon^{AB}\mu_B = -\epsilon^{BA}\mu_B
\end{eqnarray}
and
\begin{eqnarray}
\xi_A\phi^A = (\epsilon_{BA}\xi^B)\phi^A =-\epsilon_{AB}\xi^B\phi^A =-\xi^B\phi_B.
\end{eqnarray}
Finally, we shall denote the tensors obtained from $\epsilon_{AB}$ and $\epsilon^{AB}$ via complex conjugation by $\epsilon_{A^\prime B^\prime}$ and $\epsilon^{A^\prime B^\prime}$\footnote{We shall omit the bar above such tensors as in \cite{stewart}. Every other complex conjugation will be followed by a primed index and a bar over the kernel letter.}, respectively, and use them to lower and raise primed indices following the same rules as before.

If $o^A$ and $\iota^A$ are two linearly independent spinors satisfying $o_B\iota^B = \epsilon_{AB}o^A \iota^B =1$, we call $(o^A,\iota^A)$ a {\it spin basis} or a {\it spin frame}. For completeness, the relation between the spin basis and the tensors $\epsilon_{AB}$ and $\epsilon^{AB}$ are given by
\begin{eqnarray}
\epsilon^{AB} = o^A\iota^B -o^B\iota^A,\ \ \epsilon_{AB} = o_A\iota_B -o_B\iota_A
\end{eqnarray}
and it is also valid for primed objects, i.e. 
\begin{eqnarray}
\epsilon^{A^\prime B^\prime} = \bar{o}^{A^\prime}\bar{\iota}^{B^\prime} -\bar{o}^{B^\prime}\bar{\iota}^{A^\prime},\ \ \epsilon_{A^\prime B^\prime} = \bar{o}_{A^\prime}\bar{\iota}_{B^\prime} -\bar{o}_{B^\prime}\bar{\iota}_{A^\prime}.
\end{eqnarray}

\subsection{Vectors}
A spinor $\tau$ (or spinorial tensor) is called {\it hermitian} whenever $\bar{\tau}=\tau$. The fact to the matter is that the set of hermitian spinorial tensors of the type $(1,0;1,0)$ is a four-dimensional vector space over $\mathbb{R}$ that is usually (pointwise and smoothly) identified with $T_p\mathcal{M}$, the tangent space at a point in a four-dimensional manifold (to learn how this identification is done throughout $\mathcal{M}$ see \cite{stewart,penroserindler}, for instance). Analogously, the set of hermitian spinors of the type $(0,1;0,1)$ is to be (pointwise and smoothly) identified with $T^*_p\mathcal{M}$. The objects connecting spinor space and spacetime, at each point, are the so-called {\it Infeld-van der Waerden symbols} $\sigma_a^{\ \ AA^\prime}$ and its inverse $\sigma^a_{\ \ AA^\prime}$. For instance, the relation between world-vectors and spinors is given by
\begin{eqnarray}
v^a = \sigma^a_{\ \ AA^\prime}V^{AA^\prime} 
\end{eqnarray}
and for world-dual vectors 
\begin{eqnarray}
\omega_a = \sigma_a^{\ \ AA^\prime}\omega_{AA^\prime}.
\end{eqnarray}
For  general world-tensors one should apply the rule above for every upper and lower index. Since the Infeld-van der Waerden symbols represent an isomorphism, we often write expressions such as $V_{a} \sim V_{AA^\prime}$ to represent a world-vector in terms of its spinor counterpart.

For us, the most useful multivalent hermitian spinors are
\begin{eqnarray}\label{spinmetrics}
g_{ABA^\prime B^\prime}=\epsilon_{AB}\epsilon_{A^\prime B^\prime},\ \ g^{ABA^\prime B^\prime}=\epsilon^{AB}\epsilon^{A^\prime B^\prime}.
\end{eqnarray}
Their respective world-tensor is the metric and their relation is given by
\begin{eqnarray}\label{spinmetrics}
g_{ab} =\epsilon_{AB}\epsilon_{A^\prime B^\prime}\sigma_a^{\ \ AA^\prime}\sigma_b^{\ \ BB^\prime},\ \ g^{ab} =\epsilon^{AB}\epsilon^{A^\prime B^\prime}\sigma^a_{\ \ AA^\prime}\sigma^b_{\ \ BB^\prime}.
\end{eqnarray}
The main goal of section \ref{spindisf} is to investigate if one could, in the presence of two metric tensors, have a simultaneous representation of both metrics analogous to Eqs. \ref{spinmetrics}. 

The reader should verify that, given a spinor $\xi^A$, then $V^{AA^\prime} = \xi^A \bar{\xi}^{A\prime}$ represents a null vector of $T_p\mathcal{M}$. That is the reason one usually thinks of a spinor as a square root of a null-like vector. 

Finally, given a spin basis $(o^A,\iota^A)$ we can construct four linearly independent null-like world-vectors $l^a,n^a,m^a$ and $\bar{m}^a$ represented by
\begin{eqnarray}
l^a \sim l^{AA^\prime} = o^A\bar{o}^{A^\prime},\ \ n^a\sim n^{AA^\prime} = \iota^A\bar{\iota}^{A^\prime}\label{spin1}\\
m^a\sim m^{AA^\prime} = o^A\bar{\iota}^{A^\prime},\ \ \bar{m}^a\sim \bar{m}^{AA^\prime} = \iota^A\bar{o}^{A^\prime}\label{spin2}.
\end{eqnarray}
Every spin basis, therefore, defines a Newman-Penrose tetrad. A Newman-Penrose tetrad satisfy
\begin{eqnarray}\label{nptetradcontractions}
l^a n_a =l_a n^a = - m_a \bar{m}^a =-m^a \bar{m}_a=1,
\end{eqnarray}
while the other contractions vanish.

\subsection{Disformal transformations -- Spinor version}\label{spindisf}
We now turn the spinor formalism behind disformal metrics. We have seen that the spacetime metric is given by
\begin{eqnarray}
g_{ab}=\epsilon_{AB}\epsilon_{A^{\prime}B^{\prime}}\sigma_{a}^{\ \ AA^{\prime}}\sigma_{b}^{\ \ BB^{\prime}}.
\end{eqnarray}
Let us modify the Infeld-van der Waerden symbols in a certain way to obtain a new metric, this time conformal to $g_{ab}$.

If one defines $\hat{\sigma}_{a}^{\ \ AA^{\prime}}=\sqrt{\alpha}\sigma_{a}^{\ \ AA^{\prime}}$ and $\hat{g}_{ab}$ as
\begin{eqnarray}
\hat{g}_{ab}=\epsilon_{AB}\epsilon_{A^{\prime}B^{\prime}}\hat{\sigma}_{a}^{\ \ AA^{\prime}}\hat{\sigma}_{b}^{\ \ BB^{\prime}}.
\end{eqnarray}
It is easy to verify that $\hat{g}_{ab} = \alpha g_{ab}$, i.e.,  $\hat{g}$ is conformal to $g$. An alternative manner to define the conformal spacetime metric is to redefine the $\epsilon$-spinor $\epsilon_{AB}$ instead
of the Infeld-van der Waerden symbols, that is
\begin{eqnarray}
 \hat{g}_{ab}= \hat{\epsilon}_{AB} \hat{\epsilon}_{A^{\prime}B^{\prime}} \sigma_{a}^{\ \ AA^{\prime}} \sigma_{b}^{\ \ BB^{\prime}},
\end{eqnarray}
yielding the conformal $\epsilon$-spinor $\hat{\epsilon}_{AB}=\sqrt{\alpha}\epsilon_{AB}$. It is known that there exists a modification of the spinor basis $(o^{A},\iota^{A})$ producing the conformal $\epsilon$-spinor \cite{stewart}.

Let us see what happens when we try to repeat the procedure above for disformal metrics. Consider the modified Infeld-van der Waerden symbols given by
\begin{eqnarray}
\hat{\sigma}_{a}^{\ \ AA^{\prime}}= \sqrt{\alpha}\sigma_{a}^{\ \ AA^{\prime}} + \frac{\beta}{2\sqrt{\alpha}}V_{a}V^{AA^{\prime}},
\end{eqnarray}
where $V^a\sim V^{AA^{\prime}} = \xi^A \bar{\xi}^{A^\prime}$ is a fixed null-like vector field\footnote{Given any null world-vector $V^a\sim V^{AA^\prime}$, there exists a spinor $\xi^A$ such that $V^{AA^\prime}= \pm \xi^A \bar{\xi}^{A^\prime}$. In what follows there is no loss in generality if one assumes $V^{AA^{\prime}} = \xi^A \bar{\xi}^{A^\prime}$.}.
Therefore,
\begin{eqnarray}
\hat{g}_{ab}&=&\epsilon_{AB}\epsilon_{A^{\prime}B^{\prime}}\hat{\sigma}_{a}^{\ \ AA^{\prime}}\hat{\sigma}_{b}^{\ \ BB^{\prime}}\nonumber\\%
&=&\epsilon_{AB}\epsilon_{A^{\prime}B^{\prime}}\left[ \sqrt{\alpha}\sigma_{a}^{\ \ AA^{\prime}} + \frac{\beta}{2\sqrt{\alpha}}V_{a}V^{AA^{\prime}} \right]\left[ \sqrt{\alpha}\sigma_{b}^{\ \ BB^{\prime}} + \frac{\beta}{2\sqrt{\alpha}}V_{b}V^{BB^{\prime}} \right]\nonumber\\
&=&\alpha\epsilon_{AB}\epsilon_{A^{\prime}B^{\prime}}\sigma_{a}^{\ \ AA^{\prime}}\sigma_{b}^{\ \ BB^{\prime}} + \frac{\beta}{2}\epsilon_{AB}\epsilon_{A^{\prime}B^{\prime}}\sigma_{a}^{\ \ AA^{\prime}}V_{b}V^{BB^{\prime}}\\
&+& \frac{\beta}{2}\epsilon_{AB}\epsilon_{A^{\prime}B^{\prime}}\sigma_{b}^{\ \ BB^{\prime}}V_{a}V^{AA^{\prime}} + \frac{\beta^2}{4\alpha}\epsilon_{AB}\epsilon_{A^{\prime}B^{\prime}}V_{a}V_{b}V^{AA^{\prime}}V^{BB^{\prime}}.\nonumber
\end{eqnarray}
Let us simplify this equation by analyzing it term by term. The first term is given by
\begin{eqnarray}
\alpha\epsilon_{AB}\epsilon_{A^{\prime}B^{\prime}}\hat{\sigma}_{a}^{\ \ AA^{\prime}}\hat{\sigma}_{b}^{\ \ BB^{\prime}} = \alpha g_{ab}.
\end{eqnarray}
The second and third terms should be combined, after some index replacements, to yield 
\begin{eqnarray}
\frac{\beta}{2}\epsilon_{AB}\epsilon_{A^{\prime}B^{\prime}}\sigma_{a}^{\ \ AA^{\prime}}V_{b}V^{BB^{\prime}}+\frac{\beta}{2}\epsilon_{AB}\epsilon_{A^{\prime}B^{\prime}}\sigma_{b}^{\ \ BB^{\prime}}V_{a}V^{AA^{\prime}}=\nonumber\\
\beta\epsilon_{AB}\epsilon_{A^{\prime}B^{\prime}}\sigma_{a}^{\ \ AA^{\prime}}V_{b}V^{BB^{\prime}}.
\end{eqnarray}
But, 
\begin{eqnarray}
\epsilon_{AB}\epsilon_{A^{\prime}B^{\prime}}V^{BB^{\prime}}= \epsilon_{AB}\epsilon_{A^{\prime}B^{\prime}}\xi^B\bar{\xi}^{B^{\prime}}=(-\xi_{A})(-\bar{\xi}_{A^\prime})=\xi_{A}\bar{\xi}_{A^\prime}=V_{AA^\prime}
\end{eqnarray}
and $\sigma_{a}^{\ \ AA^{\prime}}V_{AA^\prime} = V_a $. In a nutshell, the second and third terms combined yield
\begin{eqnarray}
\beta\epsilon_{AB}\epsilon_{A^{\prime}B^{\prime}}\sigma_{a}^{\ \ AA^{\prime}}V_{b}V^{BB^{\prime}} = \beta V_a V_b.
\end{eqnarray}
Finally, the fourth term is given by
\begin{eqnarray}
\frac{\beta^2}{4\alpha}\epsilon_{AB}\epsilon_{A^{\prime}B^{\prime}}V_{a}V_{b}V^{AA^{\prime}}V^{BB^{\prime}} =\frac{\beta^2}{4\alpha}\epsilon_{AB}\epsilon_{A^{\prime}B^{\prime}}V_{a}V_{b}(\xi^{A}\bar{\xi}^{A^\prime})(\xi^{B}\bar{\xi}^{B^\prime})=0,
\end{eqnarray}
since $\epsilon_{AB}\xi^{A}\xi^{B}=0$.

In conclusion, we have shown that
\begin{eqnarray}
\hat{g}_{ab}= \epsilon_{AB} \epsilon_{A^{\prime}B^{\prime}} \hat{\sigma}_{a}^{\ \ AA^{\prime}} \hat{\sigma}_{b}^{\ \ BB^{\prime}} =\alpha g_{ab} + \beta V_a V_b.
\end{eqnarray}

Let us attempt to find this same decomposition, this time changing the $\epsilon$-spinor and keeping the Infeld-van der Waerden symbols, i.e.,
\begin{eqnarray}
\hat{g}_{ab}&=&\left[\alpha \epsilon_{AB} \epsilon_{A^\prime B^\prime}  + \beta (\xi_A \xi_B)(\bar{\xi}_{A^\prime}\bar{\xi}_{B^\prime})\right]\sigma_a^{\ \ AA^\prime} \sigma_b^{\ \ BB^\prime}\nonumber\\
&=&\hat{\epsilon}_{AB} \hat{\epsilon}_{A^\prime B^\prime} \sigma_a^{\ \ AA^\prime} \sigma_b^{\ \ BB^\prime}.
\end{eqnarray} 
The equation above shows that if there exists such an $\hat{\epsilon}_{AB} $ satisfying $\hat{g}_{ab}=\hat{\epsilon}_{AB} \hat{\epsilon}_{A^\prime B^\prime} \sigma_a^{\ \ AA^\prime} \sigma_b^{\ \ BB^\prime}$, then it must satisfy
\begin{eqnarray}
\hat{\epsilon}_{AB} \hat{\epsilon}_{A^\prime B^\prime} = \alpha \epsilon_{AB} \epsilon_{A^\prime B^\prime}  + \beta (\xi_A \xi_B)(\bar{\xi}_{A^\prime}\bar{\xi}_{B^\prime}),
\end{eqnarray}
and such an $\hat{\epsilon}_{AB}$ must be skew-symmetric. Furthermore, if $\beta=0$ one should recover the conformal transformation of the $\epsilon$-spinor. Therefore,
\begin{eqnarray}
\hat{\epsilon}_{AB} = \sqrt{\alpha}\epsilon_{AB} + \sqrt{\beta}T_{AB}
\end{eqnarray}
where $T_{AB}=-T_{BA}$, yielding
\begin{eqnarray}
\hat{\epsilon}_{AB} \hat{\epsilon}_{A^\prime B^\prime} &=& \alpha\epsilon_{AB}\epsilon_{A^\prime B^\prime} + \sqrt{\alpha\beta}(\epsilon_{AB}\bar{T}_{A^\prime B^\prime}+\epsilon_{A^\prime B^\prime}T_{AB})\nonumber\\
&+& \beta T_{AB}\bar{T}_{A^\prime B^\prime}.
\end{eqnarray}
Comparing this result to the desired $\hat{\epsilon}_{AB}$ one finds that, to fulfill the aforementioned conditions,
\begin{eqnarray}
\left\{
\begin{array}{c}
\beta T_{AB}\bar{T}_{A^\prime B^\prime}=\beta (\xi_A \xi_B)(\bar{\xi}_{A^\prime}\bar{\xi}_{B^\prime})\\[10pt]
\sqrt{\alpha\beta}(\epsilon_{AB}\bar{T}_{A^\prime B^\prime}+\epsilon_{A^\prime B^\prime}T_{AB})=0.
\end{array}
\right.
\end{eqnarray}
That would imply $T_{AB}=\sqrt{\beta}\xi_A \xi_B$ and that $\epsilon_{AB}\bar{T}_{A^\prime B^\prime}=-\epsilon_{A^\prime B^\prime}T_{AB}$. The first problem is that the resulting $T_{AB}$ fails to be skew-symmetric, which is one of the conditions we have established above. Not only that, the second condition says that the multivalent spinor $X_{ABA^\prime B^\prime}\equiv\epsilon_{AB}\bar{\xi}_{A^\prime}\bar{\xi}_{B^\prime}$ is skew-hermitian, i.e. $X_{ABA^\prime B^\prime}= -\bar{X}_{A^\prime B^\prime AB}$. It suffices to verify that in a basis $X_{ABA^\prime B^\prime}\neq -\bar{X}_{A^\prime B^\prime AB}$. For instance, apply $\xi^A \eta^B \bar{\eta}^{A^\prime}\bar{\eta}^{B^\prime}$ on  $X_{ABA^\prime B^\prime} +\bar{X}_{A^\prime B^\prime AB}=0$, for $\xi_{A}\eta^A \neq 0$, to check that the left side is not zero. Hence, both conditions fail to be true. It means that there is no transformation of the spin basis, except when it is conformal, that produces a spacetime disformal metric, since no transformation could produce any viable $\hat{\epsilon}_{AB}$. {\it But could there be any manner to perform such a decomposition?} Before answering this question let us see the time-like case. We could use a similar approach and define
\begin{eqnarray}
\hat{\sigma}_a^{\ \ AA^\prime} = \sqrt{\alpha}\sigma_{a}^{\ \ AA^{\prime}} + (\sqrt{\alpha+\beta}-\sqrt{\alpha})V_{a}V^{AA^{\prime}}
\end{eqnarray}
and obtain the time-like disformal metric via $\hat{g}_{ab}= \epsilon_{AB} \epsilon_{A^{\prime}B^{\prime}} \hat{\sigma}_{a}^{\ \ AA^{\prime}} \hat{\sigma}_{b}^{\ \ BB^{\prime}}$. The second step would be try to find the decomposition in terms of the $\epsilon$-spinor as before. But one should not repeat the steps above to show that there is no decomposition in terms of a modified $\epsilon$-spinor. The fact to the matter is that space of skew-symmetric tensors of type $(0,2;0,0)$ over $S$ is one-dimensional and therefore any $\hat{\epsilon}_{AB}$ should be always a multiple of the the original $\epsilon$-spinor. Unless the transformation is purely conformal, the spinor structure cannot be associated with two different spacetime metric tensors and one of the spacetime metrics should be seen merely as a tensor constructed in terms of the other. In a nutshell we have

\begin{theorem}\label{thm1}
Given a spacetime $(M,g_{ab})$ with an underlying spinor structure, a transformation of the metric tensor $g_{ab} \mapsto \widehat{g}_{ab}$ is compatible with the spinor structure if, and only if, the transformation is conformal.
\end{theorem}
\begin{proof}
For any $\Sigma_{AB}$, skew-symmetric spinorial tensor of type $(0,2;0,0)$ there exist $\chi_A,\zeta_A \in S^*$ such that
\begin{eqnarray}\label{Sigma}
\Sigma_{AB}= \chi_A \zeta_B - \chi_B\zeta_A.
\end{eqnarray}
Wrinting $\chi_A = a o_A + c \iota_A$ and $\zeta_A = b o_A + d \iota_A$ we see that Eq. (\ref{Sigma}) reduces to
\begin{eqnarray}
\Sigma_{AB}= (ad-cb)\epsilon_{AB}.\\\nonumber
\end{eqnarray}
This is a proof that the set of spinorial tensors of type $(0,2;0,0)$ is one-dimensional. Therefore any two metrics constructed with \eqref{spinmetrics} will differ by a conformal factor.\\
\end{proof}
It is important to emphasize that the set of spinorial tensors of type $(0,2;0,0)$ is one-dimensional over $\mathbb{C}$, but a complex multiple of $\varepsilon_{AB}$ will lead to torsion \cite{penroserindler}.

\begin{remark}
This result was already expected: A spinor structure on $(\mathcal{M},g_{ab})$ will determine the null vector fields of $(\mathcal{M},g_{ab})$. Therefore, changes in the metric tensor will, in general, change the causal character of vector fields. That is to say that a spinor structure over a sapcetime is compatible only with conformal transformations of the metric tensor. 
Therefore, to transform the spacetime metric one should use the objects that establish the connection between spinor space and spacetime. The Infeld-van der Waerden symbols are, then, the spinorial counterpart of the disformal transformation we have seen in section \ref{disformasse}.
\end{remark}

\section*{Acknowledgements}
IPL was financed by the Coordena\c{c}\~ao de Aperfei\c{c}oamento de Pessoal de N\'ivel Superior - Brasil (CAPES) - Finance Code 001. GGC would like to thank L. C. B. da Silva and J. A. M. Gondim for useful discussions and the National Council for Scientific and Technological Development (CNPq-Brazil) for financing this project with the post-doctoral grant No. 167485/2017-2.



\end{document}